\documentclass[runningheads,envcountsame,a4paper,final]{llncs}

\usepackage{amssymb,amsmath,mathrsfs}
\usepackage[small,basic]{complexity}
\usepackage{xspace}
\usepackage{mathtools}
\usepackage{thmtools}

\usepackage{thm-restate}
\usepackage{setspace}

\usepackage[dvipsnames,table]{xcolor}
\definecolor{ruby}{RGB}{116,0,1}
\definecolor{emerald}{RGB}{26,121,42}
\definecolor{topaz}{RGB}{236,185,57}
\definecolor{sapphire}{RGB}{14,26,164}
\usepackage{tikz}

\usepackage{bm}
\usepackage{enumerate}

\usepackage{hyperref}
\hypersetup{
	final,
	colorlinks,
	 linkcolor={sapphire},
	 citecolor={emerald},
	 urlcolor={ruby}
}

\usepackage[capitalise,nameinlink]{cleveref}
\newcommand{\appref}[1]{\hyperref[#1]{Appendix~\ref*{#1}}}

\usepackage{listings} 
\lstdefinelanguage{Sage}[]{Python}
{morekeywords={True,False,sage,singular},
sensitive=true}
\definecolor{dblackcolor}{rgb}{0.0,0.0,0.0}
\definecolor{dbluecolor}{rgb}{.01,.02,0.7}
\definecolor{dredcolor}{rgb}{0.8,0,0}
\definecolor{dgraycolor}{rgb}{0.30,0.3,0.30}

\spnewtheorem{nclam}{Claim}{\bfseries}{\itshape}
\crefname{nclam}{Claim}{Claims}

\spnewtheorem{remark}{Remark}{\bfseries}{\rmfamily} 

\newcommand{\seq}[3][{}]{\langle #2 \rangle_{#3}^{#1}}

\newclass{\EXPTIME}{EXPTIME}
\newclass{\PTIME}{PTIME}

\renewcommand{\C}{\mathbb{C}}
\newcommand{\Q}{\mathbb{Q}}

\newcommand{\Z}{\mathbb{Z}}
\newcommand{\N}{\mathbb{N}}

\newcommand{\Oc}{\mathcal{O}}

\renewcommand{\K}{\mathbb{K}}

\DeclareMathOperator{\GL}{GL}

\DeclareMathOperator{\PSL}{PSL}
\let\SL\relax \DeclareMathOperator{\SL}{SL}

\DeclareMathOperator{\Id}{Id}

\DeclareMathOperator{\ut}{U}
\DeclareMathOperator{\fg}{F}
\DeclareMathOperator{\sg}{S}

\usepackage[plain]{algorithm}
\usepackage{algorithmicx,algpseudocode}

\newcommand*\wc{{ }\cdot{ }}

\newcolumntype{L}[1]{>{\centering\arraybackslash}m{#1}}

\providecommand*{\eu}%
{\ensuremath{\mathrm{e}}}
\providecommand*{\iu}%
{\ensuremath{\mathrm{i}}}

\DeclareMathSymbol{\shortminus}{\mathbin}{AMSa}{"39}
\newcommand{\medminus}{\scalebox{0.6}[0.7]{\(-\)}}
\newcommand{\minus}{\mathchoice{-}{-}{\medminus}{\shortminus}}

  \usepackage{orcidlink}
  \renewcommand{\orcidID}[1]{\orcidlink{#1}}

\usepackage[T1]{fontenc}
\usepackage[utf8]{inputenc} \DeclareUnicodeCharacter{2212}{-}

\begin{document}
\title{On Word Representations and Embeddings in Complex Matrices
}

\author{Paul C. Bell\inst{1}\orcidID{0000-0003-2620-635X} \and
George Kenison\inst{2}\orcidID{0000-0002-7661-7061} \and 
Reino Niskanen\inst{1}\orcidID{0000-0002-2210-1481} \and 
Igor Potapov\inst{3}\orcidID{0000-0002-7192-7853} \and
Pavel Semukhin\inst{1}\orcidID{0000-0002-7547-6391}}%

\titlerunning{On Word Representations and Embeddings in Complex Matrices}
\authorrunning{P.C.~Bell et al.}

\institute{Liverpool John Moores University, UK\\
\email{\{p.c.bell,r.niskanen,p.semukhin\}@ljmu.ac.uk}
\and
Maynooth University, Ireland
\email{george.kenison@mu.ie}\\
\and
University of Liverpool, UK 
\email{potapov@liverpool.ac.uk}
}
\maketitle

\begin{abstract}
Embeddings of word structures into matrix semigroups provide a natural bridge between combinatorics on words and linear algebra. However, low-dimensional matrix semigroups impose strong structural restrictions on possible embeddings. 
Certain finitely generated groups admit faithful representations in \(\SL(2,\C)\) and other similar matrix groups.
On the other hand, it is known that the product of two free semigroups on two generators cannot be embedded into the $2\times2$ complex matrices.
In this paper  we study embeddings of word structures into low-dimensional matrix semigroups over the complex numbers and develop new techniques for constructing word representations of the Euclidean Bianchi groups. These representations provide a symbolic framework and a natural first step towards analysing fundamental decision problems in $2\times2$ matrix semigroups.
\end{abstract}

\section{Introduction}
Matrix products are one of the most fundamental operations in mathematics. Originally introduced as a compact way to represent and solve systems of linear equations, matrix products later came to play an essential role in many fields, with applications ranging from engineering and control theory to modern data science and large-scale information systems such as search ranking algorithms.

The study of embeddings of word structures into matrices connects combinatorics on words with linear algebra. By representing words or generators as matrices, questions about concatenation of words and computational models operating on words can be translated into questions about matrix multiplication. Thus, solutions to problems in combinatorics on words and formal language theory can employ algebraic and analytic tools from matrix theory. Moreover, embeddings into integer or complex matrices provide a concrete and well-understood algebraic framework in which structural properties of (semi)groups can be analysed, while also enabling the transfer of results between the theory of words and the theory of matrix semigroups.

Observe that it is not always possible to embed word structures within certain matrix classes or dimensions. This observation clarifies the expressive power of small matrix semigroups, showing that some algebraic structures generated by words cannot be faithfully represented in restricted matrix domains. 
Further, it helps determine which techniques from linear algebra can be applied to word problems, and which problems require different approaches.

The origins of the connection between combinatorial group theory and linear representations go back to the work of 
Nielsen in the 1920s on free groups \cite{wilhelm1966combinatorial}. These techniques provided an early structural understanding of free groups and influenced approaches to representing their elements by algebraic objects. Later, %
Magnus showed that free groups admit faithful representations into groups of matrices over certain rings of formal power series. Interest in such embeddings grew further with the development of algorithmic problems in algebra and theoretical computer science, when researchers began investigating whether problems concerning words and formal languages could be translated into algebraic problems involving matrices and vice versa \cite{CHK99,CN12,dong2024metabelian,KNP18,Paterson70,potapov2019vector}.

Recent work has also highlighted the importance of decision problems for matrix semigroups over complex numbers. In particular, the first decidability results for the identity and group problems in the complex Heisenberg group 
$\textrm{H}(3,\C)$ demonstrate that matrix semigroups with non-trivial structural constraints require new number theoretical techniques  \cite{BellNPS23} and new group theoretic techniques for algorithmic analysis \cite{Dong24}.
At the same time, many fundamental questions for low-dimensional matrices over the complex numbers remain open. A central example concerns embeddings of word semigroups into the ${2\times 2}$ complex matrices. 
It is known that the product of two free semigroups on two generators cannot be embedded into $\C^{2\times 2}$~\cite{CHK99}.
On the other hand, 
deciding whether a finitely generated torsion-free group embeds in \(\SL(2,\C)\) (see~\cite{Button02012016}) has deep connections to  questions in geometric group theory.
These results indicate that the expressive power of 
${2\times 2}$ complex matrices lies at a delicate boundary---while they allow rich algebraic representations, they also impose strong structural limitations on possible embeddings.

\begin{table}[htb]
\begin{center}
\begin{tabular}{c|c|L{4.1cm}|L{3cm}|L{2.5cm}}
 \(\nexists\) & \(\sg(\Sigma_1)\) & \(\sg(\Sigma_2)\) & \(\fg(\Sigma_1)\) & \(\fg(\Sigma_2)\) \\
\hline
\(\{\varepsilon\}\) & ? & ? & ? & \(\ut(n,\C)\), \(n\geq1\) (Prop.~\ref{thm:NoIntoUCn})\\
\hline
\(\sg(\Sigma_1)\) & ? & ? & ? & \textcolor{red}{\(\SL(2,\Oc_d)\)} (Prop.~\ref{prop:NonExistingEmbeddings}) \\
\hline
\(\sg(\Sigma_2)\) & {\cellcolor[gray]{0.7}} & \(\C^{2\times2}\)\cite{CHK99}, \(\SL(3,\Z)\)\cite{KNP18}  & \textcolor{red}{\(\Oc_d^{2\times2}\)} (Prop.~\ref{prop:NonExistingEmbeddings}) & \textcolor{red}{\(\Oc_d^{3\times3}\)} (Prop.~\ref{prop:NonExistingEmbeddings}) \\
\hline
\(\fg(\Sigma_1)\) & {\cellcolor[gray]{0.7}} & {\cellcolor[gray]{0.7}} & ? & \textcolor{red}{\(\Oc_d^{2\times2}\)} (Prop.~\ref{prop:NonExistingEmbeddings}) \\
\hline
\(\fg(\Sigma_2)\) & {\cellcolor[gray]{0.7}} & {\cellcolor[gray]{0.7}} & {\cellcolor[gray]{0.7}} & \(\Z^{3\times3}\)\cite{KNP18} \\
\end{tabular}
\end{center}
\caption{\label{table:notexists}
The state of the art for the non-existence of embeddings from pairs of words into different matrix semigroups. Entries in  red 
are new to the present paper.}
\end{table}

This paper explores the boundary between combinatorics on words and matrix semigroups, aiming to characterise which word structures admit faithful low-dimensional matrix  representations over complex numbers and develops new techniques to find word representations for the Euclidean Bianchi groups (\cref{thm:wordPicard}), an essential first step to study fundamental decision problems in matrix semigroups within their symbolic representations, like  membership, freeness, and vector reachability in $2\times 2$ matrix semigroups.
\cref{table:notexists} summarises non-existence results for embeddings in different settings, where 
 \(\sg(\Sigma_1)\) and \(\sg(\Sigma_2)\) denote free semigroups over unary and binary alphabets, while \(\fg(\Sigma_1)\) and \(\fg(\Sigma_2)\) denote free groups over unary and binary group alphabets.
The entries are read as ``column'' \(\times\) ``row''.
The grey background indicates symmetrical cases that do not require consideration, e.g., \(\sg(\Sigma_2)\times\sg(\Sigma_1)=\sg(\Sigma_1)\times\sg(\Sigma_2)\).
\cref{table:exists} details existing embeddings.%

\begin{table}[thb]
\begin{center}
\begin{tabular}{c|L{2cm}|L{2.9cm}|L{3cm}|L{2.9cm}}
 \(\exists\) & \(\sg(\Sigma_1)\) & \(\sg(\Sigma_2)\) & \(\fg(\Sigma_1)\) & \(\fg(\Sigma_2)\) \\
\hline
\(\{\varepsilon\}\) & & \(\ut(2, \N)\) (folklore) & \(U(2,\Z)\) (folklore), \(\Q\) (folklore) & \(\C^{2\times2}\) (folklore), \(\Z^{2\times2}\) (folklore)\\
\hline
\(\sg(\Sigma_1)\) & \textcolor{blue}{\(\ut(2,\N)\)} (Prop.~\ref{prop:ExistingEmbeddings}) & \textcolor{blue}{\(\N^{2\times2}\)} (Prop.~\ref{prop:ExistingEmbeddings}) & ? & \(\Z^{2\times2}\) (Prop.~\ref{prop:ExistingEmbeddings}) \\
\hline
\(\sg(\Sigma_2)\) & {\cellcolor[gray]{0.7}} & \({\rm SL}(3,\Q)\)~\cite{KNP18}, U\((3,\N)\)~\cite{Paterson70} & \textcolor{blue}{\(\Q^{2\times2}\)}, \(\Z^{3\times3}\) (Prop.~\ref{prop:ExistingEmbeddings}) & ? \\
\hline
\(\fg(\Sigma_1)\) & {\cellcolor[gray]{0.7}} & {\cellcolor[gray]{0.7}} & \(\Q\) (folklore) & \(\Q^{2\times2}\) (Prop.~\ref{prop:ExistingEmbeddings}) \\
\hline
\(\fg(\Sigma_2)\) & {\cellcolor[gray]{0.7}} & {\cellcolor[gray]{0.7}} & {\cellcolor[gray]{0.7}} & \(\Z^{4\times4}\), \(\mathbb{H}(\Q)^{2\times 2}\) \cite{BP08} \\
\end{tabular}
\end{center}
\caption{\label{table:exists}The state of the art for the existence of embeddings from pairs of words into different matrix semigroups. 
Entries in {blue} are straightforward extensions of results in the literature.
}
\end{table}

The paper is structured as follows.
In the next section, we introduce necessary notations and preliminaries.
In \cref{sec:wordrep}, we prove that an element of a Euclidean Bianchi group admits a canonical word representation.
In \cref{sec:embedpairs}, we investigate the border between the existence and non-existence of injective morphisms from pairs of word (semi)groups to various matrix semigroups.

\section{Preliminaries} \label{sec:preliminaries}
\subsection{Words, semigroups, and groups}
Given an \emph{alphabet} $\Sigma = \{a_1,a_2, \ldots, a_m\}$, a finite \emph{word} $u$ is a finite sequence of letters, $u=u_1u_2\cdots u_n$, where $u_i\in\Sigma$.
We denote the free semigroup of finite words over \(\Sigma\) by \(\sg(\Sigma)\) where \(\sg(\Sigma) = \langle a_1, \ldots, a_m \rangle\), with \(\langle a_1, \ldots, a_m \rangle\) denoting the semigroup \emph{generated} by \(a_1, \ldots a_m\).
The \emph{empty word} is denoted by $\varepsilon$.
The length of a finite word~$u$ is denoted by $|u|$ and $|\varepsilon|=0$.
We denote by \(\fg(\Sigma)\) the free group \(\langle a_1, a_1^{-1}, \ldots, a_m, a_m^{-1}\rangle\), i.e., the free group generated by the letters \(a_i\) and their inverses \(a_i^{-1}\).
Naturally, \(a_ia_i^{-1}=a_i^{-1}a_i=\varepsilon\).
The elements of \(\fg(\Sigma)\) are all \emph{reduced} words over $\Sigma$, i.e., words not containing $a_ia_i^{-1}$ or $a_i^{-1}a_i$ as a subword.
In this context, we call \(\Sigma\) a finite \emph{group alphabet}, i.e., an alphabet with an involution.
The multiplication of two elements (reduced words) \(u,v\in \fg(\Sigma)\) corresponds to the unique reduced word of the concatenation \(uv\).
The length of an element of \(\fg(\Sigma)\) is given by the length of its reduced word representation.

Let \(\varphi\) be a mapping from \(\sg(\Sigma)\) into \(\K^{n\times n}\), where \(\K\) is an algebraic structure (e.g., \(\N,\Z,\Q,\C\)) and \(n\geq1\).
That is, a mapping of semigroup words into \(n\)-by-\(n\) matrices with elements from \(\K\).
We say that \(\varphi\) is an \emph{injective morphism} or an \emph{embedding} if \(\varphi(u)\varphi(v)=\varphi(uv)\) for every \(u,v\in\sg(\Sigma)\) and if \(\varphi(u)=\varphi(v)\) implies that \(u=v\).
We define \emph{group embeddings} over the free group \(\fg(\Sigma)\) analogously.

The following proposition is folklore.
It allows us to focus on small alphabets, namely unary or binary alphabets, in the sequel.
\begin{proposition}
Let \(\Sigma_k = \{a_1, a_2, \ldots, a_k\}\) be an alphabet for any \(k \geq 2\). 
    \begin{itemize}
        \item If there exists an embedding \(\sigma:\sg(\Sigma_2)\to \mathbb{F}^{n \times n}\), then there exists an embedding \(\sigma:\sg(\Sigma_k)\to \mathbb{F}^{n \times n}\) for any \(k \geq 2\).
        \item If there exists a group embedding \(\sigma:\fg(\Sigma_2)\to \mathbb{F}^{n \times n}\), then there exists a group embedding \(\sigma:\fg(\Sigma_k)\to \mathbb{F}^{n \times n}\) for any \(k \geq 2\).
    \end{itemize} 
\end{proposition}

\subsection{Integer rings, Euclidean domains, and quadratic fields}
We shall assume some familiarity with concepts from algebraic number theory (cf.~\cite{stewart2016algebraic}).
We refer the reader to the background material on properties of certain imaginary quadratic fields in 
Appendix~\ref*{app:QuadraticFields}.

Let \(\overline{\Q}\) denote the field of algebraic numbers.
Recall that a number is \emph{algebraic} if it is the root of a non-zero polynomial \(p\in\Z[x]\) with integer coefficients.
Further, a number is an \emph{algebraic integer} if it is the root of a non-zero monic polynomial \(p\in\Z[x]\).
We can effectively represent and compute algebraic numbers~\cite{cohen2013course}.

An \emph{integral domain} is a non-zero commutative ring such that the product of any two non-zero elements is itself non-zero.
An integral domain \(E\) is Euclidean (or a \emph{Euclidean domain}) if there exists a function \(g \colon E \setminus\{0\}\to \Z_{\ge 0}\) such that for every \(x,y\in E\setminus\{0\}\),
\begin{itemize}
    \item \(g(xy) \ge g(x)\);
    \item if \(x,y\in E\) and \(y\neq 0\).  The there exist \(q,r\in E\) with \(a=qb+r\) and either \(r=0\) or \(g(r)<g(b)\).
\end{itemize}
A function \(g\) with such properties is a \emph{Euclidean function}~\cite{dubois1958divisionalgorithms,hardy2008introduction}.

The group of \emph{units} \(R^\times\) of a ring \(R\) is the subset  of elements that possess a multiplicative inverse element in the ring.
 Two elements \(a,b\in R\) of an integral domain are \emph{associates} if there exists a unit \(u\in R\) such that \(a = bu\).

Recall that a number field \(\mathbb{K}\) is \emph{quadratic} (or a \emph{quadratic field}) if there is a square-free integer \(-d\in\Z\) such that \(\mathbb{K}=\mathbb{Q}(\sqrt{-d})\).  If, in addition, \(d\in\N\), then the field \(\mathbb{Q}(\sqrt{-d})\) is an \emph{imaginary quadratic field}.  In the work that follows, we denote by \(\Oc_d\) the ring of algebraic integers in \(\mathbb{Q}(\sqrt{-d})\).
For \(d\in\Z\) square-free, the quadratic field \(\Q(\sqrt{-d})\) admits a  \emph{field norm} \(N\colon \Q(\sqrt{d}) \to \Q\)~\cite{stewart2016algebraic}.
This norm is multiplicative, so that \(N(\alpha \beta) = N(\alpha)N(\beta)\), for algebraic integers \(\alpha\) we have \(N(\alpha)\in\Z\), and for imaginary quadratic fields we have \(N(\alpha) = |\alpha|^2\).
For the ring of algebraic integers in a number field, the units are precisely those integers \(u\) for which \(N(u)=1\).

The next result precisely characterises imaginary quadratic fields \(\Q(\sqrt{-d})\) whose integer rings \(\Oc_d\) are Euclidean domains.
\begin{lemma}[{\cite[Proposition 8.9]{hatcher2022numbers}}]
\label{lem:euclideandomain}
The ring of integers \(\Oc_d\) for the imaginary quadratic field \(\Q(\sqrt{-d})\) is a Euclidean domain if and only if \(d\in\{1,2,3,7,11\}\). For each such \(d\), the field norm on \(\Q(\sqrt{-d})\) is Euclidean. 
\end{lemma}
\subsection{Matrix groups}
We shall assume some familiarity with (semi)groups of square matrices such as \(\GL(2,R) = \{M\in R^{2\times2} : \det(M)\in R^\times\}\), the \emph{General Linear group} of \(2\times 2\) matrices with entries in the ring \(R\) whose inverses also have entries in the ring \(R\),
the \emph{Special Linear group} \(\SL(2,R) = \{M\in \GL(2,R) : \det(M)=1\}\), and the quotient group \(\SL(2,R) /\{\pm \Id_2\} \eqqcolon \PSL(2,R)\) called the \emph{Projective Special Linear group}.

By definition, we associate to each element \(m\in\PSL(2,R)\) a set \(\{M, -M\}\) of two matrices
in \(\SL(2,R)\). Henceforth we employ a slight abuse of notation and write \(m = \pm M\), or choose either matrix \(M\) or \(-M\) to represent \(m\).
Intuitively, one can take \(\PSL(2,R)\) as \(\SL(2,R)\) by ignoring the sign.

In the literature, the \emph{Bianchi groups} are the groups \(\PSL(2,\Oc_d)\) where \(\Oc_d\) is a ring of integers of an imaginary quadratic field \(\Q(\sqrt{-d})\)~\cite{bianchi1892}. In the sequel, we are interested in the \emph{Euclidean Bianchi groups} \(\PSL(2,\Oc_d)\) where \(d=1,2,3,7,11\).
The naming convention follows from the result in \cref{lem:euclideandomain}: of the imaginary quadratic number rings, only \(\mathcal{O}_1, \mathcal{O}_2, \mathcal{O}_3, \mathcal{O}_7\), and \(\mathcal{O}_{11}\) are Euclidean domains.
Observe that \(\mathcal{O}_1\) is the Gaussian ring \(\Z[\iu]\). The group \(\PSL(2,\Z[\iu])\) is commonly called the Picard group~\cite{fine1987picard,fine1989bianchi}.

\section{Word Representation Procedure}
\label{sec:wordrep}

It is well-known that each of the Euclidean Bianchi groups is finitely generated~\cite[Theorem 4.3.1]{fine1989bianchi}. The groups \(\PSL(2,\Oc_1)\) and \(\PSL(2,\Oc_3)\) admit group presentations with four generators each, whilst \(\PSL(2,\Oc_2)\), \(\PSL(2,\Oc_7)\), and \(\PSL(2,\Oc_{11})\) are each generated by three elements.
The matrices associated to each of these generators (commonly denoted by \(a\), \(t\), \(u\), and \(\ell\)) are:
    \begin{equation*}
        A = \begin{pmatrix}
        0 & -1 \\ 1 & \phantom{-}0
    \end{pmatrix}, \quad
    T = \begin{pmatrix}
        1 & 1 \\ 0 & 1
    \end{pmatrix}, \quad
    U = \begin{pmatrix}
        1 & \zeta \\ 0 & 1
    \end{pmatrix}.
    \end{equation*}
The top-right entry \(\zeta\) of \(U\) depends on the specific matrix group \(\PSL(2,\Oc_d)\); we give further details in Appendix~\ref*{app:QuadraticFields}.
The fourth matrix \(L\) is given as follows,
\begin{equation*}
    L = \begin{dcases} 
        \begin{pmatrix}
        \iu & \phantom{-}0 \\
        0 & -\iu
    \end{pmatrix} & \text{if } d=1, \text{ and } \\
    \begin{pmatrix}
        \omega^2 & 0 \\
        0 & \omega
    \end{pmatrix} & \text{if } d=3 \text{ where } \omega = - \tfrac{1}{2} + \tfrac{\sqrt{3}\iu}{2}. \\
    \end{dcases}
\end{equation*}

The following theorem generalises the procedure in \cite[Lemma 3.1]{bell2017identity} that generates word representations for elements of \(\PSL(2,\Z)\). 
\begin{restatable}{theorem}{wordBianchi}
\label{thm:wordPicard}
    There is a procedure that, given an element \(M\) in a Euclidean Bianchi group \(\PSL(2,\Oc_d)\), outputs a word representation for \(M\)
    of the form
        \begin{align*}
              (L^\epsilon T^{p_0} U^{q_0}) \cdot A T^{p_k} U^{q_k} \cdot A T^{p_{k-1}} U^{q_{k-1}} \cdots A T^{p_1} U^{q_1} & \text{ if } d\in\{1,3\}, \text{ and} \\
              (T^{p_0} U^{q_0}) \cdot A T^{p_k} U^{q_k} \cdot A T^{p_{k-1}} U^{q_{k-1}} \cdots A T^{p_1} U^{q_1} & \text{ if } d \in\{2,7,11\}.
        \end{align*}
    Here \(\epsilon\in\{0,1,2\}\), the exponent pairs \(p_\ell, q_\ell\in\Z\) each satisfy \(|p_\ell + q_\ell\omega|^2 \le \|M\|\) where \(\|M\|\coloneqq \max_{1\le i,j \le 2} |M_{ij}|^2\), {and} \(k < 1-\log_{\kappa(d)} \|M\|\) where \(\kappa(d)\) is the Euclidean minimum of \(\Oc_d\) 
    (see Appendix~\ref*{app:QuadraticFields}). 
    Moreover, this procedure runs in time polynomial in 
        \( -\log_{\kappa(d)} \|M\|\).
\end{restatable}

\begin{proof}
    For the ease of presentation, we give the procedure for elements of the Picard group \(\PSL(2,\Oc_1)\) (=\(\PSL(2,\Z[\iu])\)) here. \textit{Mutatis mutandis,} analogous arguments for the remaining Euclidean Bianchi groups are given in 
    Appendix~\ref*{app:wordBianchi}. 
    There are two parts to our proof: the construction of the word representation and the polynomial runtime.\\
{\bf Construction of the word representation.}
Suppose that \(M = \begin{psmallmatrix} \alpha & \beta \\ \gamma & \delta \end{psmallmatrix} \in \PSL(2,\Z[\iu])\).  
Our first step is to construct an element \(H\in\PSL(2,\Z[\iu])\) such that \(MH\) is upper-triangular (thus reducing the representation problem to that of representing an upper-triangular element). The second step constructs the word representation of an upper-triangular element.

We begin with the first step.
In the case that \(\gamma=0\), we choose \(H = \textrm{Id}_2\).
We continue under the assumption that \(\gamma \neq 0\).
We claim that there is an \(H_1\in\PSL(2,\Z[\iu])\)  such that \(MH_1 = \begin{psmallmatrix} \alpha_1 & \beta_1 \\ \gamma_1 & \delta_1\end{psmallmatrix}\) where \(N(\gamma_1) < N(\gamma)\).
(Here \(N\) is the field norm on \(\Q(\iu)\) with \(N(x+y\iu)=x^2 + y^2\) for \(x+y\iu\in\Z[\iu]\).)
Since \(\Z[\iu]\) equipped with the field norm \(N\) is a Euclidean domain,
we can write 
\(\delta = -\theta_1\gamma + \gamma_1\) with \(\theta_1,\gamma_1\in\Z[\iu]\) and so the inequality
    \(N(\gamma_1) = N(\theta_1\gamma+\delta) < N(\gamma)\)
follows.
Let us write \(\theta_1= -p_1 -q_1\iu\in\Z[\iu]\) in terms of the integral basis \(\{1,\iu\}\), then
\begin{multline*}
    M U^{-q_1} T^{-p_1} A = 
    \begin{pmatrix}
        \alpha & \beta \\ \gamma & \delta
    \end{pmatrix}
    \begin{pmatrix}
        1 & \iu \\ 0 & 1
    \end{pmatrix}^{-q_1}
    \begin{pmatrix}
        1 & 1 \\ 0 & 1
    \end{pmatrix}^{-p_1}
    \begin{pmatrix}
        0 & -1 \\ 1 & \phantom{-}0
    \end{pmatrix}
     \\
     =
    \begin{pmatrix}
        \alpha\ & \theta_1 \alpha  + \beta \\
        \gamma\ & \theta_1 \gamma  + \delta \\
    \end{pmatrix}
    \begin{pmatrix}
        0 & -1 \\ 1 & \phantom{-}0
    \end{pmatrix}
    =
    \begin{pmatrix}
       \theta_1 \alpha  + \beta\  & -\alpha \\
       \theta_1 \gamma  + \delta\  & -\gamma
    \end{pmatrix}
    \eqqcolon
    \begin{pmatrix} \alpha_1 & \beta_1 \\ \gamma_1 & \delta_1\end{pmatrix}.
\end{multline*}
Since \(N(\gamma_1) < N(\gamma)\),
the element \(H_1\coloneqq  U^{-q_1} T^{-p_1} A\) has the claimed properties.

We loop the above construction in order to generate a sequence of matrices of the form \(M H_1 \cdots H_\ell\).
Let \(\gamma_\ell\) be the bottom-left entry of matrix \(M H_1\cdots H_\ell\).
We repeat this process until we obtain a matrix \(M H_1 \cdots H_k\) that satisfies the condition \(\gamma_k=0\).

The following observations guarantee that the above process both terminates and does so correctly.
First, \(N(\gamma_\ell) = |\gamma_\ell|^2\in\Z_{\ge 0}\) for each \(\ell\) since \(\gamma_\ell\in\Z[\iu]\).
Second, the sequence of these norms is strictly decreasing (and so \(N(\gamma_{\ell+1}) \le  N(\gamma_{\ell}) - 1\)
for \(1\le \ell \le k-1\)).
Hence there exists \(k\in\N\) such that \(N(\gamma_k)=0\), from which we deduce that \(\gamma_k=0\).
Thus we have constructed an element \(H_1\cdots H_k \eqqcolon H\in\PSL(2,\Z[\iu])\) such that \(MH\) is upper-triangular. %
For our second step, 
consider
    \begin{equation*}
        MH= \begin{pmatrix} \alpha_k & \beta_k \\ 0 & \delta_k \end{pmatrix} \eqqcolon %
\begin{pmatrix} \rho & \sigma \\ 0 & \tau \end{pmatrix}
    \end{equation*} with \(\rho,\sigma,\tau\in \Z[\iu]\).
Since \(1=\det(MH)=\rho \tau\), we deduce that \(\rho,\tau\in \Z[\iu]^\times = \{\pm1, \pm\iu\}\) and, moreover,  \(\overline{\tau} = \tau^{-1} = \rho\).
The possible pairs \((\rho,\tau)\) ensure that
 \begin{equation*}
     MH =   
    \begin{pmatrix}
        \iu & \phantom{-}0 \\ 0 & -\iu
    \end{pmatrix}^{\epsilon} 
    \begin{pmatrix}
        1 & \sigma' \\ 0 & 1
    \end{pmatrix}
=  
    \begin{pmatrix}
        \iu & \phantom{-}0 \\ 0 & -\iu
    \end{pmatrix}^{\epsilon} 
    \begin{pmatrix}
        1 & 1 \\ 0 & 1
    \end{pmatrix}^{p_0}
       \begin{pmatrix}
        1 & \iu \\ 0 & 1
    \end{pmatrix}^{q_0}
=
     L^\epsilon T^{p_0} U^{q_0}
 \end{equation*}
    for some \(\epsilon\in\{0,1\}\) and \(\sigma' = p_0+q_0\iu\in\Z[\iu]\) (an associate of \(\sigma\)).  

Taken together, these two procedural steps construct representations for elements \(H\in\PSL(2,\Z[\iu])\) and \(L^\epsilon T^{p_0} U^{q_0}\) for which \(M= L^\epsilon T^{p_0} U^{q_0} H^{-1}\). 
Using these products, we can output a word representation of the desired form for a given \(M\).\\
{\bf Polynomial runtime.}
It is useful to introduce the following notations.
For \(1\le \ell \le k\), we shall write
\begin{equation} \label{eq:update}
    M_{\ell} \coloneqq MH_1\cdots H_{\ell} =
    \begin{pmatrix}
    \theta_\ell \alpha_{\ell-1} + \beta_{\ell-1} & -\alpha_{\ell-1} \\ \theta_\ell \gamma_{\ell-1} + \delta_{\ell-1} & - \gamma_{\ell-1}
    \end{pmatrix}
    \eqqcolon 
    \begin{pmatrix}
        \alpha_\ell & \beta_\ell \\ \gamma_\ell & \delta_\ell
    \end{pmatrix}.
\end{equation}

The update matrix \(M_\ell\) is obtained by performing a Euclidean division in the Gaussian integers \(\Oc_1=\Z[\iu]\).
Recall that a division \(a= qb + r\) in the Gaussian integers satisfies \(N(r)< N(b)\) where \(N\) is the associated field norm.
In fact, working in \(\Z[\iu]\) we have the tighter upper bound \(N(r) < \tfrac{1}{2} N(b)\) since \(\tfrac{1}{2}\) is the Euclidean minimum of \(\Z[\iu]\).
Further background on Euclidean minima can be found in 
Appendix~\ref*{app:QuadraticFields}. 

For \(\ell<k\), repeated application of the Euclidean minimum gives
\begin{equation*}
   1\le  %
   N(\gamma_\ell) < \tfrac{1}{2} N(\gamma_{\ell-1}) 
   < \cdots < \tfrac{1}{2^{\ell-1}} N(\gamma_1) < \tfrac{1}{2^\ell} N(\gamma) = \tfrac{1}{2^\ell} |\gamma|^2 \le \tfrac{1}{2^\ell} \|M\|.
\end{equation*}
Taking logarithms, we obtain the upper bound \(\log_{2} \|M\| > \ell\), from which the desired inequality \(1+\log_{2} \|M\| > k\) follows.

We now exhibit bounds on the matrix norms \(\seq[k]{\|M_\ell\|}{\ell=1}\).
The determinant condition on \(M_{\ell}\) tells us that 
\begin{multline*}
    \alpha_\ell \delta_{\ell} - \beta_\ell\gamma_\ell =-(\theta_\ell \alpha_{\ell-1} + \beta_{\ell-1})\gamma_{\ell-1}  + (\theta_\ell \gamma_{\ell-1} + \delta_{\ell-1})\alpha_{\ell-1}  \\
    =  -(\theta_\ell \alpha_{\ell-1} + \beta_{\ell-1})\gamma_{\ell-1}  + \gamma_{\ell}\alpha_{\ell-1} =1.
\end{multline*}
For \(\ell-1<k\), we have that \(\gamma_{\ell-1}\neq 0\) and so
    \begin{multline} \label{eq:normbound}
        |\alpha_{\ell}|^2 = N(\alpha_\ell) = N(\theta_\ell\alpha_{\ell-1} + \beta_\ell) = N\!\left(\frac{\gamma_\ell\alpha_{\ell-1} - 1}{\gamma_{\ell-1}}\right) \le \\ \frac{N(\gamma_\ell)N(\alpha_{\ell-1})}{N(\gamma_{\ell-1})} + \frac{1}{N(\gamma_{\ell-1})} \le \frac{1}{2}N(\alpha_{\ell-1}) + 1.
    \end{multline}
In~\eqref{eq:normbound}, the rightmost inequality follows from the division \(-\delta_{\ell-1}= \theta_{\ell} \gamma_{\ell-1} + \gamma_{\ell}\).

We make the following claim.
\begin{nclam} \label{claim:inequalityPicard}
     For \(M_{\ell-1}\) and \(M_{\ell}\) %
     in \(\PSL(2,\Z[\iu])\) as above,
     we necessarily have that \(\|M_{\ell}\| \le \|M_{\ell-1}\|\).
\end{nclam}

All that remains is to bound  the sizes of the integer exponents in the word representation.
For \(\ell\le k\), we bound the quotients \(\theta_\ell \coloneqq -p_\ell - q_\ell\iu\) as follows, %
\begin{multline*}
    |p_\ell|^2 + |q_\ell|^2 = N(\theta_\ell) = N\!\left( \frac{\delta_{\ell-1}+\gamma_\ell}{\gamma_{\ell-1}} \right) \\ \le N(\delta_{l-1}) + \frac{N(\gamma_\ell)}{N(\gamma_{\ell-1})} \le \|M_{\ell-1}\| + \tfrac{1}{2} \le \|M\| + \tfrac{1}{2}.
\end{multline*}
In the above line, the rightmost inequality follows from \cref{claim:inequalityPicard}.
Since \(N(\theta_\ell), \|M\|\in\Z\), we have \(N(\theta_\ell)\le \|M\|\) and so \(|p_\ell|^2,|q_\ell|^2\le \|M\|\).
We consider the matrix \(MH\) in order to bound \(p_0\) and \(q_0\). Once again, it is clear that
    \begin{equation*}
        |p_0|^2, |q_0|^2 \le |\sigma|^2 \le \max_{1\le i,j \le 2} |(MH)_{ij}|^2 = \|M_k\| \le \|M\|.
    \end{equation*}
The above observations bound the number of iterations in the initial looping procedure as well as the sizes of the exponents in the word representation. Taken together with standard results on the division algorithm, we obtain the stated polynomial runtime. \qed
\end{proof}

\begin{proof}[Proof of \cref{claim:inequalityPicard}]
By \eqref{eq:update}, we have \(|\beta_{\ell}|^2,|\delta_\ell|^2, |\gamma_\ell|^2 \le \|M_{\ell-1}\|\).
Thus \(\|M_{\ell} \| \le \|M_{\ell-1}\|\) unless 
\(|\alpha_\ell|^2 > \|M_{\ell-1}\|\).
Let us assume, for a contradiction, that \(|\alpha_\ell|^2 > \|M_{\ell-1}\|\).
By \eqref{eq:normbound}, it follows that 
    \begin{equation*}
        \|M_{\ell-1}\| < \|M_\ell\| = N(\alpha_\ell) \le \frac{1}{2}N(\alpha_{\ell-1}) + 1 < \frac{1}{2} \|M_{\ell-1}\| + 1,
    \end{equation*}
and so we deduce that \(\|M_{\ell-1}\| < 2\).
Thus \((M_{\ell-1})_{ij}\in\{0,\pm 1, \pm\iu\}\) for each \((i,j)\).

    We first argue that it is not possible that each entry of \(M_{\ell-1}\) is non-zero; for otherwise, each entry is a unit and the determinant condition \(\alpha_{\ell-1}\delta_{\ell-1} - \beta_{\ell-1}\gamma_{\ell-1}=1\) is not satisfied by any tuple of units in \(\Z[\iu]^\times\). 
    The determinant condition also ensures that \(M_{\ell-1}\neq 0_{2\times 2}\).
    Thus \(M_{\ell-1}\) takes one of the following forms
        \begin{equation*}
            \begin{pmatrix}  0 & \beta_{\ell-1} \\ \gamma_{\ell-1} & \delta_{\ell-1}      \end{pmatrix}, \,
            \begin{pmatrix}  0 & \beta_{\ell-1} \\ \gamma_{\ell-1} & 0       \end{pmatrix}, \,
            \begin{pmatrix}  \alpha_{\ell-1} & 0 \\ \gamma_{\ell-1} & \delta_{\ell-1}      \end{pmatrix}, \, \text{or} \,
            \begin{pmatrix}  \alpha_{\ell-1} & \beta_{\ell-1} \\ \gamma_{\ell-1} & 0       \end{pmatrix}
        \end{equation*}
    where an entry not denoted by \(0\) is a unit (i.e., in \(\{\pm 1, \pm \iu\} = \Z[\iu]^\times\)).
    In the first two forms, where \(\alpha_{\ell-1}=0\), it is clear that the update \eqref{eq:update} ensures that \(|\alpha_{\ell}|^2 = |\beta_{\ell-1}|^2 > 0\) and so the required inequality trivially holds. 
    Likewise, the inequality holds for the third form since, by the update \eqref{eq:update}, \(|\alpha_\ell|^2 = |\alpha_{\ell-1}|^2\).
    All that remains is to treat the fourth form. 
    The update to the fourth form gives
    \(\begin{psmallmatrix}
        \alpha_{\ell} & \beta_{\ell} \\ \gamma_{\ell} & \delta_{\ell}
    \end{psmallmatrix}
    = 
        \begin{psmallmatrix}
            \beta_{\ell-1} & \minus \alpha_{\ell-1} \\ 0 & \minus \gamma_{\ell-1}
        \end{psmallmatrix}
    \)
    because \(0 =\theta_\ell \gamma_{\ell-1}+ \gamma_{\ell} = \theta_{\ell} \gamma_{\ell-1}\) implies that \(\theta_{\ell}=0\).
    Thus \(|\alpha_\ell|^2 = |\beta_{\ell-1}|^2 = 1 = |\alpha_{\ell-1}|^2\).
    Thus the required inequality holds. \qed
\end{proof}

As seen in~\cite{bell2017identity}, a procedure that generates a word representation of an element in  \(\PSL(2,\Z)\) implies a procedures that generates the word representation of an element in \(\SL(2,\Z)\).
In a similar vein, we have the following corollary.%
\begin{corollary} \label{cor:wordModular}
    For each of the groups \(\SL(2,\Oc_d)\) with \(d\in\{1,2,3,7,11\}\), there is a procedure that, given an element \(M\in\SL(2,\Oc_d)\) outputs a word representation for \(M\) in terms of the generators of \(\SL(2,\Oc_d)\).
\end{corollary}

\section{Word Embeddings Into Matrix (Semi)Groups}\label{sec:embedpairs}

Let us first recall two known group embeddings from the literature (cf.~\cite{BP08,CHK99}).
\begin{proposition}\label{prop:classicalGroupEmbedding1}
   Let \(\Sigma_2 = \{a, b\}\). Then \(\varphi:\fg(\Sigma_2)\to \mathbb{Z}^{2 \times 2}\) defined by 
   \begin{equation*}
       \varphi(a) = \begin{pmatrix} 1 & 2 \\ 0 & 1 \end{pmatrix}, \quad
   \varphi(b) = \begin{pmatrix} 1 & 0 \\ 2 & 1 \end{pmatrix}
   \end{equation*}
   is an embedding.
\end{proposition}

\begin{proposition}\label{prop:classicalGroupEmbedding2}
   Let \(\Sigma_2 = \{a, b\}\). Then \(\varphi:\fg(\Sigma_2)\to \mathbb{C}^{2 \times 2}\) defined by
   \begin{equation*}
       \varphi(a) = \begin{pmatrix} \frac{3}{5} + \frac{4}{5}\iu & 0 \\[0.25em] 0 & \frac{3}{5} - \frac{4}{5}\iu \end{pmatrix}, \quad 
   \varphi(b) = \begin{pmatrix} \phantom{-}\frac{3}{5} & \frac{4}{5} \\[0.25em] -\frac{4}{5} & \frac{3}{5} \end{pmatrix}
   \end{equation*}
   is an embedding.
\end{proposition}

For both of the above embeddings, all four sets 
\(\{x\in \fg(\Sigma_2) : (\varphi(x))_{ij} \neq 0\}\) (where \(1\le i,j \le 2\)) are non-empty, i.e., the codomains of both embeddings contain matrices with non-zero values in each of the four entries.
For comparison, let \(\psi\) denote Paterson's classical embedding for pairs of binary semigroup words into \(\N^{3\times3}\)~\cite{Paterson70}.  Then \(\{x\in \sg(\Sigma_2)\times \sg(\Sigma_2) : (\psi(x))_{21} \neq 0\}\) is empty, i.e., Paterson's embedding uses only the elements on and above the main diagonal.
Further, it is known that one cannot embed \(\fg(\Sigma_2)\) into the upper triangular matrices of any dimension. That is to say,

\begin{proposition}\label{thm:NoIntoUCn}
There is no embedding from \(\fg(\Sigma_2)\) into \(\ut(n,\C)\) for any $n$. Here \(\ut(\mathbb{C},n)\) is the group of \(n\times n\) upper triangular matrices with complex entries.
\end{proposition}
The result in \cref{thm:NoIntoUCn} (and the generalisation that one cannot embed \(\fg(\Sigma_2)\) into any solvable group) follows from standard results in group theory involving normal series~\cite[Chapter 5]{Rotman1995}.

We briefly sketch the proof of \cref{thm:NoIntoUCn}.
Recall that \(\ut(n,\C)\) is solvable and that every subgroup of a solvable group is itself solvable~\cite[Theorem 5.15]{Rotman1995}.  Further,  no subgroup of a solvable group is isomorphic to \(\fg(\Sigma_2)\) (cf.~\cite[Chapter 5.3]{Rotman1995}).
The existence of an embedding \(\varphi\colon \fg(\Sigma_2) \to \ut(n,\C)\) contradicts the above observation. that \(\fg(\Sigma_2)\) is isomorphic to some subgroup of \(\ut(n,\C)\).

The analogous result (to~\cref{thm:NoIntoUCn}) for semigroups is as follows: one cannot embed \(\sg(\Sigma_2)\) into any nilpotent group. This follows from the observation that the number of words in \(\sg(\Sigma_2)\) with length at most \(n\) grows exponentially, whilst nilpotent groups exhibit polynomial growth~\cite{Milnor1968}.

Motivated by the above existence and non-existence results for embeddings,
it is interesting to investigate the (algebraic) complexity of the word structures one can embed into restricted groups and semigroups of low-dimensional matrices.
The next two propositions provide a summary of our results.
For the rest of the section, we assume that \(\Sigma_2=\{a,b\}\) and \(\Sigma_1=\{c\}\).

First, we consider some existence results.
The details of the embeddings build on the well-known embeddings from the literature and can be found in 
 Appendix~\ref*{app:nonexistence}. 
\begin{restatable}{proposition}{embeddingexistence}
\label{prop:ExistingEmbeddings} \hspace{0.2em} \hfil
\vspace{-0.5em}
\begin{itemize}
 \item There exists an embedding from \(\fg(\Sigma_1)\times\fg(\Sigma_1)\) into \(\ut(2,\N)\).
 \item There exists an embedding from \(\fg(\Sigma_2)\times\fg(\Sigma_1)\) into \(\Q^{2\times2}\).
 \item There exists an embedding from \(\sg(\Sigma_2)\times\fg(\Sigma_1)\) into \(\Z^{3\times3}\).
 \item There exists an embedding from \(\fg(\Sigma_2)\times\sg(\Sigma_1)\) into \(\Z^{2\times2}\).
 \end{itemize}
\end{restatable}

We now turn our attention to non-existence results in similar settings. 

\begin{restatable}{proposition}{embeddingnonexistence}
\label{prop:NonExistingEmbeddings}
   Let \(d\in\{1,2,3,7,11\}\).
    \vspace{-0.5em}
\begin{itemize}
\item There is no embedding from \(\sg(\Sigma_2)\times\fg(\Sigma_1)\) into \(\Oc_d^{2\times2}\).
\vspace{0.25em}\item There is no embedding from \(\fg(\Sigma_2)\times\fg(\Sigma_1)\) into \(\Oc_d^{2\times2}\).
\vspace{0.25em}\item There is no embedding from \(\fg(\Sigma_2)\times\sg(\Sigma_1)\) into \(\mathrm{SL}(2,\Oc_d)\).
\vspace{0.25em}\item There is no embedding from \(\sg(\Sigma_2)\times\sg(\Sigma_2)\) into \(\mathrm{SL}(3,\Oc_d)\).
\vspace{0.25em}\item There is no embedding from \(\fg(\Sigma_2)\times\sg(\Sigma_2)\) into \(\Oc_d^{3\times3}\).
\end{itemize}
\end{restatable}
\begin{proof}
Let us prove the first and second statements.
Proofs of the other three statements use similar approaches akin to techniques found in~\cite{CHK99,KNP18}. 
Additional details can be found in 
Appendix~\ref*{app:nonexistence}. 

Let \(\Sigma_2=\{a,b\}\) and \(\Sigma_1=\{c\}\).
Assume, for a contradiction, that there exists an embedding \(\varphi:\sg(\Sigma_2)\times\fg(\Sigma_1) \to \Oc_d^{2\times2}\).
Embedding \(\varphi\) maps the generators as follows:
\begin{align*}
    (a,\varepsilon)\mapsto \begin{pmatrix}
        a_1&a_2\\a_3&a_4
    \end{pmatrix}, \quad
    &(b,\varepsilon)\mapsto \begin{pmatrix}
        b_1&b_2\\b_3&b_4 \\
    \end{pmatrix},  \\
    (\varepsilon,c)\mapsto \begin{pmatrix}
        c_1&c_2\\c_3&c_4
    \end{pmatrix}, \quad
    &(\varepsilon,c^{-1})\mapsto \frac{1}{c_1c_4-c_2c_3}\begin{pmatrix}
        \phantom{-}c_4&-c_2\\-c_3& \phantom{-}c_1
    \end{pmatrix}.
\end{align*}
The last mapping implies that \(c_1c_4-c_2c_3 \in \Oc_d^\times\); that is, \(c_1c_4-c_2c_3\) is a unit.
Furthermore, we have the following relations:
\begin{equation}
\begin{aligned}\label{eq:rels}
    \varphi((a,\varepsilon)(b,\varepsilon))&\neq\varphi((b,\varepsilon)(a,\varepsilon)), \\
    \varphi((a,\varepsilon)(\varepsilon,c)) &= \varphi((\varepsilon,c)(a,\varepsilon)), \\
    \varphi((b,\varepsilon)(\varepsilon,c)) &= \varphi((\varepsilon,c)(b,\varepsilon)).
\end{aligned}
\end{equation}

Let us consider the products \(\varphi((a,\varepsilon)(\varepsilon,c))\) and \(\varphi((\varepsilon,c)(a,\varepsilon))\):
\begin{align}\label{eq:ac}
    \varphi((a,\varepsilon)(\varepsilon,c)) &=  \begin{pmatrix}
        a_1c_1+a_2c_3\ & a_1c_2+a_2c_4 \\ a_3c_1+a_4c_3\ & a_3c_2+a_4c_4
    \end{pmatrix}, \\
    \varphi((\varepsilon,c)(a,\varepsilon)) &= \begin{pmatrix} \label{eq:ca}
        a_1c_1+a_3c_2\ & a_2c_1+a_4c_2 \\ a_1c_3+a_3c_4\ & a_2c_3+a_4c_4
    \end{pmatrix}.
\end{align}
In order to derive a contradiction, we first assume that \(c_2=0\).
Under this assumption, the top-right corner elements in the matrices \(\varphi((a,\varepsilon)(\varepsilon,c))\) and \(\varphi((\varepsilon,c)(a,\varepsilon))\) are \(a_2c_4\) and \(a_2c_1\), respectively.
By \eqref{eq:rels}, we know these two matrices are equal and so infer that \(c_1=c_4\).
Analogously reasoning about the bottom-left corner entries, we have that \(a_1=a_4\). 
From the third relation in \eqref{eq:rels}, we similarly deduce that \(b_1=b_4\).

We now turn our attention to the bottom-right corner entries in \eqref{eq:ac} and \eqref{eq:ca}, which, from the preceding work, are \(a_1c_1\) and \(a_2c_3+a_1c_1\), respectively.
By \eqref{eq:rels}, these entries are equal.
Thus we have that either \(a_2=0\) (and analogously \(b_2=0\)) or \(c_3=0\).
In the latter case, \(\varphi(\varepsilon,c)=\begin{psmallmatrix}
    c_1&0\\0&c_1
\end{psmallmatrix}\) and \(\varphi(\varepsilon,c^{-1})=\frac{1}{c_1^2}\begin{psmallmatrix}
    c_1&0\\0&c_1
\end{psmallmatrix}\).
This implies that \(c_1\in \Oc_d^\times\).
For each \(d\), the elements of \(\Oc_d^\times\) are listed in 
\cref{lem:unitintegers}. 
It is straightforward to see that \((\varphi(\varepsilon,c))^k=\Id_2\) for some \(k\ge1\), which contradicts our assumption that \(\varphi\) is an embedding.
As an example, consider \(d=3\) and the unit \(\omega=\tfrac{1}{2}+\tfrac{\sqrt{-3}}{2}\).
Then \(\varphi(\varepsilon,c)^6=\begin{psmallmatrix}
    \omega&0\\0&\omega
\end{psmallmatrix}^6=\Id_2\).
We deduce that \(c_3\neq 0\).

Let us turn to the second subcase, where we assume that
\(a_2 = b_2=0\).
Under this assumption,
we obtain the following equality
\begin{equation*}
    \varphi((a,\varepsilon)(b,\varepsilon))= \begin{pmatrix}
        a_1b_1&0\\a_3b_1+a_1b_3\ &a_1b_1
    \end{pmatrix}
    =
    \varphi((b,\varepsilon)(a,\varepsilon)),
\end{equation*}
which contradicts the inequality in \eqref{eq:rels}.
{We derive an analogous contradiction if 
if we instead work under the assumption that \(c_3=0\) (rather than \(c_2=0\)).
}

Observe that in equations~\eqref{eq:ac} and \eqref{eq:ca}, the top-left corner elements are equal if and only if \(a_2c_3=a_3c_2\).
The preceding work showed that both \(c_2\) and \(c_3\) are non-zero.
We note that \(a_2\) and \(a_3\) (and \(b_2\) and \(b_3\)) are also non-zero.
If this was not the case, then
\begin{align}
    \varphi((a,\varepsilon)(b,\varepsilon))= \begin{pmatrix}
        a_1b_1&a_1b_2\\a_4b_3&a_4b_4
    \end{pmatrix} \quad \text{and} \quad
    \varphi((b,\varepsilon)(a,\varepsilon))= \begin{pmatrix}
        a_1b_1&a_4b_2\\a_1b_3&a_4b_4
    \end{pmatrix} \label{eq:align1}
\intertext{as well as}
    \varphi((a,\varepsilon)(\varepsilon,c)) = \begin{pmatrix}
        a_1c_1 & a_1c_2\\ a_4c_3 & a_4c_4
    \end{pmatrix} \quad \text{and} \quad
    \varphi((\varepsilon,c)(a,\varepsilon)) = \begin{pmatrix}
        a_1c_1&a_4c_2 \\ a_1c_3 & a_4c_4
    \end{pmatrix}. \label{eq:align2}
\end{align}
By \eqref{eq:rels}, the two matrices at \eqref{eq:align1} are not equal, whilst the two matrices at \eqref{eq:align2} are equal.
For both of these statements to be true,
both \(a_1\neq a_4\) and \(a_1=a_4\) have to hold simultaneously, a contradiction.

Finally, observe that the top-left corners of the products imply that
\(a_2c_3=a_3c_2\), \(b_2c_3=b_3c_2\), and \(a_2b_3\neq a_3b_2\).
Since \(a_2,a_3,b_2,b_3,c_2\) and \(c_3\) are all non-zero, we have \(\frac{a_2}{a_3}=\frac{c_2}{c_3}=\frac{b_2}{b_3}\) which contradicts the inequality \(a_2b_3\neq a_3b_2\).
Thus there is no embedding from \(\sg(\Sigma_2)\times\fg(\Sigma_1)\) into \(\Oc_d^{2\times2}\).

The second statement follows directly: if there is no embedding from a semigroup alphabet, then there is also no embedding from a group alphabet. \qed
\end{proof}

\begin{remark}
In the proof of \cref{prop:NonExistingEmbeddings}, we use the fact that the unit group \(\Oc_d^\times\) is finite.
\textit{Mutatis mutandis,} the same non-existence results hold for any integer ring with a torsion unit group.
By Dirichlet's unit theorem (cf.~\cite[Chapter 1, Theorem 7.4]{neukirch1999algebraic}), 
a number field \(\K\) has a torsion unit group (i.e., every unit has finite order) if and only if 
\(\K=\Q\) or \(\K=\Q(\sqrt{-d})\) for square-free \(d\ge 1\).
Thus the non-existence statements in \cref{prop:NonExistingEmbeddings} for embeddings into \(\Oc_d^{2\times 2}\) with \(d\in\{1,2,3,7,11\}\) also hold for the codomains \(\Z^{2\times 2}\) and \(\Oc_d^{2\times 2}\) for square-free \(d\ge 1\).
\end{remark}

\section{Conclusion and Future Directions}

In \cref{sec:wordrep}, we presented the first step towards showing the decidability of the identity problem for matrices over \(\SL(2,\Oc_d)\).
Indeed, in \cite{bell2017identity}, the word representation algorithm for \(\SL(2,\Z)\) was used to construct finite petal graphs where the existence of a short path can be checked in \(\NP\).
The construction utilised a complete, confluent, and monadic term rewriting system for the generators of \(\SL(2,\Z)\).
No such system is known for Bianchi groups.
One can obtain either a confluent or a monadic rewriting system, but this is not sufficient for the construction as the finiteness of the petal graph is no longer guaranteed.
A pertinent question is whether there exists a confluent and monadic rewriting system for \(\SL(2,\Oc_d)\) or whether one can ensure that the resulting petal graph is finite.

In \cref{sec:embedpairs}, we considered various matrix semigroups and showed that there is no way to embed pairs of words over different alphabets.
It is interesting to ask, how narrow is the boundary between existence and non-existence?
In particular, when considering \(\K^{n\times n}\), what is the smallest \(n\) and the largest \(\K\) such that the embedding no longer exists.
It is also worth noting that we considered embeddings of pairs of words.
One can extend the problem to the \(k\)-fold products, i.e., \(\sg(\Sigma)^k\) and \(\fg(\Sigma)^k\); see \cite{CK21}.

\subsubsection*{Acknowledgements.} We are grateful to the anonymous reviewers' careful readings and helpful suggestions that led to many improvements in this manuscript. George Kenison was partially supported by the FWO grants G0F5921N (Odysseus) and G023721N, and by the KU Leuven grant iBOF/23/064.
Igor Potapov was supported by the APEX Awards 2024
- APX\textbackslash R1\textbackslash241186.

\bibliography{BKNPS26}  %
\bibliographystyle{splncs04}

 \newpage
 \appendix
 \section{Background Material on Imaginary Quadratic Fields} \label[appendix]{app:QuadraticFields}

We have the following characterisation of the ring of algebraic integers in the imaginary quadratic field \(\mathbb{Q}(\sqrt{-d})\)~\cite[Chapter 3]{stewart2016algebraic}.
\begin{lemma}
Suppose that  \(d\in\Z\) is square-free. The subring of algebraic integers \(\Oc_d\) in the quadratic field \(\Q(\sqrt{-d})\) has the form \(\Oc_d = \Z[\omega] \coloneqq \{x+y\omega \mid x,y\in\Z\}\) where
    \begin{equation*}
        \omega = \begin{dcases}
            \sqrt{-d} & \text{if } -d \equiv 2, 3 \pmod{4},\,\text{and} \\
            (1 + \sqrt{-d})/2 & \text{if } -d \equiv 1 \pmod{4}.
        \end{dcases}
    \end{equation*}
\end{lemma}

The \emph{field norm} of the integer \(x+y\omega\in \mathcal{O}_d\) is given as follows:
    \begin{equation*}
        N(x+y\omega) \coloneqq \begin{dcases}
            x^2 + y^2d & \text{if \(-d\equiv 2,3 \pmod{4}\)}, \\
            x^2 + xy + y^2\left( \tfrac{1+d}{4} \right) & \text{if \(-d\equiv 1 \pmod{4}\)}.
        \end{dcases}
    \end{equation*}

\begin{lemma} \label{lem:unitintegers}
    Suppose that \(d\in\N\) is square-free.  Then the group of units \(\Oc_d^\times\) in the ring of imaginary quadratic integers \(\Oc_d\) is finite.
    Further,
     \begin{equation*}
     \Oc_d^\times = \begin{dcases}
       \{\pm 1, \pm \iu\} & \text{if} \quad d=1, \\
       \bigl\{\pm 1, \pm \tfrac{1}{2} \pm \tfrac{\sqrt{-3}}{2}, \pm \tfrac{1}{2} \mp \tfrac{\sqrt{-3}}{2} \bigr\} & \text{if} \quad d=3, \\
        \{\pm 1\} & \text{if} \quad d=2,7,11.
    \end{dcases}
    \end{equation*}
\end{lemma}

Recall the classical Euclidean algorithm for elements \(a,b\in\Z\) allows one to choose \(q,r\in\Z\) such that \(a=qb+r\) with \(|r|<|b|/2\).  The following definition abstracts the contraction factor of \(1/2\) on the absolute value of the remainder to the setting of algebraic integers (although we limit ourselves to quadratic integer rings).
Given a field \(\Q(\sqrt{-d})\), the corresponding ring of integers \(\Oc_d\), and field norm \(N\) as above, we define
    the \emph{Euclidean minimum}~\cite{lemmermeyer1995euclidean} of \(\Q(\sqrt{-d})\) by
    \begin{equation*}
        \kappa(d) \coloneqq \sup_{\alpha\in\Q(\sqrt{\minus d})} \inf_{\beta\in\Oc_d} | N(\alpha-\beta)|.
    \end{equation*}
The connection to the Euclidean algorithm is made clear by the equivalent definition
\begin{equation*}
    \kappa(d) = \inf\{ \kappa>0 \mid \forall\alpha,\beta\in\Oc_d\!\setminus\!\{0\}, \, \exists q\in \Oc_d \text{ such that } N(\alpha- q\beta)<\kappa N(\beta) \}.
\end{equation*}

Surveys on Euclidean minima in the general algebraic setting are given in~\cite{bayerfluckiger2006euclidean,lemmermeyer1995euclidean}.
For \(\Oc_d\) to be a Euclidean domain, it is sufficient that \(\kappa(d)<1\), cf.~\cite[\textsection 2.3--2.4]{conway2003quaternions} and \cite[Chapter 8]{hatcher2022numbers}).
The minima for the Euclidean domains we discuss are given below.
\begin{proposition}[{\cite[Proposition 4.2]{lemmermeyer1995euclidean}}] %
\label{prop:euclideanminima}
    For \(d\) and \(\Oc_d\) as above,
    \begin{equation*}
        \kappa(d) = \begin{dcases}
            \frac{d+1}{4} & \text{ if } d=1,2, \text{ and} \\
            \frac{(d+1)^2}{16 d} & \text{ if } d=3,7,11.
        \end{dcases}
    \end{equation*}
\end{proposition}
The minima for the five Euclidean imaginary quadratic integer rings are given in \cref{table:matrixentries}.

\setlength\extrarowheight{2pt}
\begin{table}
\begin{centering}
\rowcolors{1}{}{gray!20!white}
\begin{tabular}{c | c | c | c | p{7cm} }
 \(d\) & basis element \(\omega\)  & \(\kappa(d)\)  & \(\tfrac{1}{1-\kappa(d)}\)  & \(\bigl\{x\in \Oc_d : N(x) < \tfrac{1}{1-\kappa(d)} \bigr\}\) \\[0.2em] \hline
 1 & \(\iu\)   & \(\tfrac{1}{2}\)  & 2 & \(\{0, \pm 1, \pm \iu\}\)  \\[0.25em]
 2 & \(\sqrt{\shortminus 2}\)   & \(\tfrac{3}{4}\)  & 4  & \(\{0, \pm 1, \pm \sqrt{-2}, -1 \pm \sqrt{-2}, 1 \pm \sqrt{-2}\}\)  \\[0.25em]
 3 & \(\tfrac{1}{2} + \tfrac{\sqrt{-3}}{2}\)  & \(\tfrac{1}{3}\) & \(\tfrac{3}{2}\) & \(\{0, \pm 1, \pm\omega, \pm\overline{\omega} \}\)  \\[0.25em]
  7 & \(\tfrac{1}{2} + \tfrac{\sqrt{-7}}{2}\)  & \(\tfrac{4}{7}\) & \(\tfrac{7}{3}\) & \(\{0, \pm 1, \pm\omega, \pm\overline{\omega}\}\) \\[0.25em]
  11 & \(\tfrac{1}{2} + \tfrac{\sqrt{-11}}{2}\) & \(\tfrac{9}{11}\) & \(\tfrac{11}{2}\) & \(\{0, \pm 1, \pm2, \pm\omega, -1\pm\omega, 1\pm\omega, \pm(\omega -2)\}\) \\[0.25em]
\end{tabular}
\vspace{0.5em}
\caption{Data for the integer rings \(\Oc_d\) relevant to the Euclidean Bianchi groups \(\PSL(2,\Oc_d)\). \label{table:matrixentries}} %
\end{centering}
\end{table}
\setlength\extrarowheight{0pt}

\section{Word Representation Procedure for the Euclidean Bianchi Groups}
\label[appendix]{app:wordBianchi}

Recall the main result in \cref{sec:wordrep}.
\wordBianchi*

We previously showed that \cref{thm:wordPicard} holds for the Picard group \(\PSL(2,\Z[\iu])\).
In this section, we will complete the proof of \cref{thm:wordPicard} by establishing the word representation procedures for each of the remaining Euclidean Bianchi groups.
The reader will observe that, \textit{mutatis mutandis}, the structure and machinery is analogous to that employed in the proof for the Picard group.

\begin{proof}[Proof of \cref{thm:wordPicard} for \(\PSL(2,\Oc_d)\) with \(d=2,3,7,11\)]
    We split the proof of \cref{thm:wordPicard} into two parts: the construction of the word representation and the polynomial runtime.

\textbf{Construction of the word representation.}
Suppose that \(M = \begin{psmallmatrix} \alpha & \beta \\ \gamma & \delta \end{psmallmatrix} \in \PSL(2,\Oc_d)\).  
Our first step is to construct an element \(H\in\PSL(2,\Oc_d)\) such that \(MH\) is upper-triangular.
In the case that \(\gamma=0\), we choose \(H = \textrm{Id}_2\).

We continue under the assumption that \(\gamma \neq 0\).
We claim that there is an \(H_1\in\PSL(2,\Oc_d)\)  such that \(MH_1 = \begin{psmallmatrix} \alpha_1 & \beta_1 \\ \gamma_1 & \delta_1\end{psmallmatrix}\) where \(N(\gamma_1) < N(\gamma)\).
(Here \(N\) is the field norm on \(\Q(\sqrt{-d})\).)
Since \(\Oc_d\) equipped with the field norm \(N\) is a Euclidean domain,
we can write 
\(\delta = -\theta_1\gamma + \gamma_1\) with \(\theta,\gamma_1\in\Oc_d\), from which
    \(N(\gamma_1) = N(\theta_1\gamma+\delta) < N(\gamma)\) follows.
Let us write \(\theta_1= -p_1 -q_1\omega\in\Oc_d\) in terms of the integral basis \(\{1,\omega\}\), then
\begin{multline*}
    M U^{-q_1} T^{-p_1} A = 
    \begin{pmatrix}
        \alpha & \beta \\ \gamma & \delta
    \end{pmatrix}
    \begin{pmatrix}
        1 & \omega \\ 0 & 1
    \end{pmatrix}^{-q_1}
    \begin{pmatrix}
        1 & 1 \\ 0 & 1
    \end{pmatrix}^{-p_1}
    \begin{pmatrix}
        0 & -1 \\ 1 & \phantom{-}0
    \end{pmatrix}
     \\
     =
    \begin{pmatrix}
        \alpha\ & \theta_1 \alpha  + \beta \\
        \gamma\ & \theta_1 \gamma  + \delta \\
    \end{pmatrix}
    \begin{pmatrix}
        0 & -1 \\ 1 & \phantom{-}0
    \end{pmatrix}
    =
    \begin{pmatrix}
       \theta_1 \alpha  + \beta\ & -\alpha \\
       \theta_1 \gamma  + \delta\ & -\gamma
    \end{pmatrix}
    \eqqcolon
    \begin{pmatrix} \alpha_1 & \beta_1 \\ \gamma_1 & \delta_1\end{pmatrix}.
\end{multline*}
The constructed element \(H_1\coloneqq  U^{-p_1} T^{-q_1} A\) has the required property that \(N(\gamma_1) < N(\gamma)\). %

We loop the above construction in order to generate a sequence of matrices of the form \(M H_1 \cdots H_\ell\).
Let \(\gamma_\ell\) be the bottom-left entry of the matrix \(M H_1\cdots H_\ell\).
We repeat this loop until we reach a matrix \(M H_1 \cdots H_k\) that satisfies the condition \(\gamma_k=0\).

The following observations guarantee that the above loop both terminates and does so correctly.
First, \(N(\gamma_\ell) = |\gamma_\ell|^2\in\Z_{\ge 0}\) for each \(\ell\) since \(\gamma_\ell\in\Oc_d\).
Second, the sequence of these norms is strictly decreasing (and so \(N(\gamma_{\ell+1}) \le  N(\gamma_{\ell}) - 1\)
for \(1\le \ell \le k-1\)).
Hence there exists \(k\in\N\) such that \(N(\gamma_k)=0\), from which we deduce that \(\gamma_k=0\).
Thus we have constructed an element \(H_1\cdots H_k \eqqcolon H\in\PSL(2,\Oc_d)\) such that \(MH\) is upper-triangular. %
This ends our first step, which 
has reduced the representation problem to that of representing an upper-triangular element.

For our second step, 
consider
    \begin{equation*}
        MH= \begin{pmatrix} \alpha_k & \beta_k \\ 0 & \delta_k \end{pmatrix} \eqqcolon %
\begin{pmatrix} \rho & \sigma \\ 0 & \tau \end{pmatrix}
    \end{equation*} with \(\rho,\sigma,\tau\in \Oc_d\).
Since \(1=\det(MH)=\rho \tau\), we deduce that \(\rho,\tau\in \Oc_d^\times\) and, moreover,  \(\overline{\tau} = \tau^{-1} = \rho\).
The remainder of this step is broken into two cases that depend on the multiplicative group of units \(\Oc_d^\times\) (see \cref{lem:unitintegers}).
\begin{itemize}
    \item Suppose that \(d=3\). Then for any \(\rho\in \Oc_d^\times = \{\pm 1, \pm \omega, \pm \omega^2\}\) where \(\omega = -\tfrac{1}{2} + \tfrac{\sqrt{3}\iu}{2}\), we pick \(\tau = \overline{\rho}\).
    From the possible pairings, we deduce that
    \begin{equation*}
        MH =   
    \begin{pmatrix}
        \omega^2 & 0 \\ 0 & \omega
    \end{pmatrix}^{\epsilon} 
    \begin{pmatrix}
        1 & \sigma' \\ 0 & 1
    \end{pmatrix}
    =
    \begin{pmatrix}
        \omega^2 & 0 \\ 0 & \omega
    \end{pmatrix}^{\epsilon} 
    \begin{pmatrix}
        1 & 1 \\ 0 & 1
    \end{pmatrix}^{p_0}
       \begin{pmatrix}
        1 & \omega \\ 0 & 1
    \end{pmatrix}^{q_0}
=
      L^\epsilon T^{p_0} U^{q_0}
    \end{equation*}
    where \(\epsilon\in\{0,1,2\}\) and \(\sigma' = p_0+q_0\omega\in\Z[\omega]=\Oc_3\) is an associate of \(\sigma\).

    \item Suppose that \(d\in\{2,7,11\}\).  
    Then \(\Oc_d^\times = \{\pm 1\}\). 
    In this case, our only options are \(\rho= \tau= \pm 1\). %
    Thus 
        \begin{equation*}
            MH =   %
    \begin{pmatrix}
        1 & \sigma' \\ 0 & 1
    \end{pmatrix}
    =
    \begin{pmatrix}
        1 & 1 \\ 0 & 1
    \end{pmatrix}^{p_0}
    \begin{pmatrix}
        1 & \omega \\ 0 & 1
    \end{pmatrix}^{q_0}
=
      T^{p_0} U^{q_0}
        \end{equation*}
where \(\pm\sigma = \sigma' = p_0 + q_0\omega\in\Oc_d\). %
\end{itemize}
Thus ends the second step.

Taken together, the above work constructs representations for elements \(H\) and \(T^{p_0} U^{q_0}\) (or \(L^\epsilon T^{p_0} U^{q_0}\)) for which \(M = T^{p_0} U^{q_0} H^{-1}\) (or \(M= L^\epsilon T^{p_0} U^{q_0} H^{-1}\)). 
Using these products, we can output a word representation of the desired form for a given \(M\).

\textbf{Polynomial runtime.}
We shall employ the same notations \(\|\wc \|\) for the matrix norm and \(M_{\ell}\) for matrix products as before. %
Recall that \(\kappa(d)\) is the Euclidean minima on \(\mathbb{Q}(\sqrt{-d})\) (see \cref{prop:euclideanminima} and \cref{table:matrixentries}).
For \(\ell<k\), repeated application of the Euclidean minima gives
\begin{multline*}
   1\le  |\gamma_\ell|^2 = N(\gamma_\ell) < \kappa(d) N(\gamma_{\ell-1}) < \kappa(d)^2 N(\gamma_{\ell-2}) < \cdots \\ < \kappa(d)^{\ell-1} N(\gamma_1) < \kappa(d)^\ell N(\gamma) = \kappa(d)^\ell |\gamma|^2 = \kappa(d)^\ell \|M\|.
\end{multline*}
Taking logarithms, we find the upper bound \(-\log_{\kappa(d)} \|M\| > \ell\).
Arguing the contrapositive, we obtain the upper bound \(-\log_{\kappa(d)} \|M\| + 1 > k\) on termination of the first step.

The determinant condition on \(M_\ell\) tells us that 
\begin{multline*}
    \alpha_\ell \delta_{\ell} - \beta_\ell\gamma_\ell =-(\theta_\ell \alpha_{\ell-1} + \beta_{\ell-1})\gamma_{\ell-1}  + (\theta_\ell \gamma_{\ell-1} + \delta_{\ell-1})\alpha_{\ell-1}  \\
    =  -(\theta_\ell \alpha_{\ell-1} + \beta_{\ell-1})\gamma_{\ell-1}  + \gamma_{\ell}\alpha_{\ell-1} =1.
\end{multline*}
For \(\ell-1<k\), we have that \(\gamma_{\ell-1}^{-1}\) is well-defined and so
    \begin{multline} \label{eq:normboundEB}
        N(\alpha_\ell) = N(\theta_\ell\alpha_{\ell-1} + \beta_\ell) = N\!\left(\frac{\gamma_\ell\alpha_{\ell-1} - 1}{\gamma_{\ell-1}}\right) \\ \le \frac{N(\gamma_\ell)N(\alpha_{\ell-1})}{N(\gamma_{\ell-1})} + \frac{1}{N(\gamma_{\ell-1})} \le \kappa(d)N(\alpha_{\ell-1}) + 1.
    \end{multline}
Here the rightmost inequality of \eqref{eq:normboundEB} follows from the Euclidean division \(-\delta_{\ell-1}= \theta_{\ell} \gamma_{\ell-1} + \gamma_{\ell}\), which gives the bound \( N(\gamma_\ell) < \kappa(d) N(\gamma_{\ell-1})\).
\begin{nclam} \label{claim:inequalityPicard2}
For each \(d\in \{2,3,7,11\}\), there are no matrices \(M_{\ell-1}\) and \(M_\ell\) (as above) for which \(\|M_{\ell}\| > \|M\|_{\ell-1}\).
\end{nclam}
The proof of \cref{claim:inequalityPicard2} is given in \cref{appendix:proofofclaim}.
All that remains is to bound  the sizes of the integer exponents in the word representation.
For \(\ell\le k\), we give the following bounds on the quotients \(\theta_\ell \coloneqq -p_\ell - q_\ell\omega\), %
\begin{multline*}
    |p_\ell + q_\ell\omega|^2 = N(\theta_\ell) = N\!\left( \frac{\delta_{\ell-1}+\gamma_\ell}{\gamma_{\ell-1}} \right) \\ \le N(\delta_{l-1}) + \frac{N(\gamma_\ell)}{N(\gamma_{\ell-1})} \le \|M_{\ell-1}\| + \kappa(d) \le \|M\| + \kappa(d).
\end{multline*}
Since \(N(\theta_\ell), \|M\|\in\Z\), we have \(N(\theta_\ell)\le \|M\|\), as desired.
Observe that the rightmost inequality follows from \cref{claim:inequalityPicard2}.
We consider the matrix \(MH\) in order to bound \(|p_0 + q_0 \omega|^2\). 
Once again, it is clear that
    \begin{equation*}
        |p_0+q_0\omega|^2 = |\sigma|^2 \le \max_{1\le i,j \le 2} |(MH)_{ij}|^2 = \|M_k\| \le \|M\|
    \end{equation*}
The above observations bound the number of iterations in the initial looping procedure as well as the sizes of the exponents in the word representation. Taken together, with standard results on the division algorithm leads us to the stated polynomial runtime. \qed
\end{proof}

\subsection{Proof of \texorpdfstring{\cref{claim:inequalityPicard2}}{the Claim} } 
\label[appendix]{appendix:proofofclaim}

In this section, we shall prove the pending cases for \cref{claim:inequalityPicard2}. 
We take this opportunity to sketch a proof outline.
For each \(d=2,3,7,11\), consider ordered pairs of matrices \((M, M')\) in \(\PSL(2,\Oc_d)\) given by
\(M = \begin{psmallmatrix} \alpha & \beta \\ \gamma & \delta \end{psmallmatrix}\) and
\(%
    M' =
    \begin{psmallmatrix}
    \theta \alpha + \beta & -\alpha \\ 
    \theta \gamma + \delta & - \gamma
    \end{psmallmatrix}
\)
(as in \eqref{eq:update})
where \(-\theta,r\in\Oc_d\) are the quotient and remainder in \(\delta = -\theta \cdot \gamma + r\), respectively.
We need to show that each such pair \((M, M')\) satisfies \(\|M\|<\|M'\|\).

The proof proceeds by contradiction, any such pair \((M,M')\) requires that
    \begin{equation} \label{eq:matrixnormbound}
        \|M\| = \max_{1\le i,j\le 2}|M_{ij}|^2 \le \frac{1}{1-\kappa(d)}
    \end{equation}
(where \(\kappa(d)\) is the Euclidean minima of \(\Oc_d\)).
Thus, to prove the claim, it is sufficient to compute and compare the norms \(\|M\|\) and \(\|M'\|\) for all such pairs \((M,M')\).
The number of matrices \(M\) that satisfy \eqref{eq:matrixnormbound} 
is too large to deal with by hand.
Thus, with computer assistance\footnote{Our computations were performed in SageMath~\cite{sage}.}, we perform an exhaustive search for said matrices and establish that \(\|M\| > \|M'\|\) holds for each candidate.
Thus ends our sketch proof.

\begin{proof}[Proof of \cref{claim:inequalityPicard2}]
As in the Gaussian integer case of \cref{thm:wordPicard}, the matrix norm satisfies \(\|M_{\ell}\| > \|M\|_{\ell-1}\) only if \(|\alpha_\ell|^2 > |\alpha_{\ell-1}|^2\). Assume, for a contradiction, that \(\|M_{\ell}\| > \|M_{\ell-1}\|\) then, by \eqref{eq:normboundEB},
    \begin{equation*}
        \|M_{\ell-1}\| < \|M_\ell\| = N(\alpha_\ell) \le \kappa(d)N(\alpha_{\ell-1}) + 1 < \kappa(d) \|M_{\ell-1}\| + 1,
    \end{equation*}
and so we deduce that \((1-\kappa(d))\|M\|_{\ell-1} < 1\). The possible entries for a matrix in \(\PSL(2,\Oc_d)\) are listed in \cref{table:matrixentries}.
Our proof continues with an exhaustive search for pairs \((M_{\ell-1}, M_\ell)\) of matrices that falsify the claim (i.e., pairs with \(\|M_{\ell-1}\| < \|M_\ell\|\)).
We implemented this search in SageMath~\cite{sage}; our code is available at~\cite{kenison2026SageEuclideanBianchi}.
Pseudocode for this search procedure is given in \cref{alg:cap}.

From our implementation and testing, we confirm that \cref{claim:inequalityPicard,claim:inequalityPicard2} hold for each \(d\in\{1,2,3,7,11\}\). \qed
\end{proof}

\begin{algorithm}
\hrulefill
\setstretch{1.35}
\caption{For \(d=1\), this procedure determines whether \cref{claim:inequalityPicard} holds.  For \(d=2,3,7,11\), this procedure determines whether \cref{claim:inequalityPicard2} holds.\label{alg:cap}}
\begin{algorithmic}[1]
\Require {\(d\in\{1,2,3,7,11\}\)}
\State \texttt{matrices} \(\leftarrow\) empty list 
\State \(S\leftarrow \bigl\{x\in \Oc_d : N(x) < \tfrac{1}{1-\kappa(d)}\bigr\}\); \Comment{set of matrix entries (see \cref{table:matrixentries})}
\For{ordered tuples \((\alpha,\beta,\gamma, \delta)\), with \(\alpha,\beta,\gamma,\delta\in S\), \(\gamma\neq 0\), and \(\alpha \delta - \beta \gamma = 1\)} 
\State \(\bigl( M, \|M\|\bigr)  \leftarrow \left( \begin{psmallmatrix} \alpha & \beta \\ \gamma & \delta \end{psmallmatrix} , \max\{|\alpha|^2, |\beta|^2, |\gamma|^2, |\delta|^2\}\right)\); \Comment{matrix and its norm}
\State \(\theta \leftarrow \textsc{Quotient}(-\delta, \gamma)\); \Comment{output quotient \(-\theta\in\Oc_d\)} %
\State \((\alpha,\beta) \leftarrow (\theta \cdot \alpha + \beta, -\alpha)\); 
\State \((\gamma, \delta) \leftarrow (\theta \cdot \gamma + \delta, -\gamma)\);
\State \(\bigl( M', \|M'\| \bigr)  \leftarrow \left(\begin{psmallmatrix} \alpha & \beta \\ \gamma & \delta \end{psmallmatrix}, \max\{|\alpha|^2, |\beta|^2, |\gamma|^2, |\delta|^2\} \right)  \); \Comment{update matrix and its norm}
\If{\(\|M\| < \|M'\|\)}
\State add \((M, M')\) to \texttt{matrices} %
\EndIf
\EndFor
\If{\texttt{matrices} is an empty list}
\State \Return {\lq\lq Claim holds for matrices in \(\PSL(2,\Oc_d)\)\rq\rq}
\Else
\State \Return {\lq\lq Claim does not hold for the pairs} \texttt{matrices}\rq\rq
\EndIf
\end{algorithmic}
\hrulefill
\end{algorithm}

In \cref{alg:cap}, the function \(\textsc{Quotient}\colon \Oc_d \times \Oc_d \to \Oc_d\) takes as input an ordered pair of algebraic integers and returns their quotient (in the ring \(\Oc_d\)).
In other words, for the ordered pair \((a,b)\in\Oc_d\times \Oc_d\), \textsc{Quotient} returns the nearest integer (in \(\Oc_d\)) to \(a/b\).
In \cref{alg:cap}, the function \(\textsc{Quotient}\colon \Oc_d \times \Oc_d \to \Oc_d\) takes as input an ordered pair of algebraic integers and returns their quotient (in the ring \(\Oc_d\)).
In other words, for the ordered pair \((a,b)\in\Oc_d\times \Oc_d\), \textsc{Quotient} returns the nearest integer (in \(\Oc_d\)) to \(a/b\).

SageMath does not automatically recognise that the integer rings \(\Oc_d\) with \(d=1,2,3,7,11\) are Euclidean domains.
Thus for integers \(a,b\in\Oc_d\), the SageMath command \verb|a.quo_rem(b)[0]| outputs the result of the division \(a/b\in\Q(\sqrt{-d})\) rather than the nearest integer in \(\Oc_d\) (the output of \textsc{Quotient}).
We circumvent this obstacle by implementing a straightforward procedure for finding the nearest algebraic integer in \(\Oc_d\)~\cite{kenison2026SageEuclideanBianchi}.

\section{Omitted Embeddings Proofs}\label[appendix]{app:nonexistence}
\embeddingexistence*
\begin{proof}
Consider mappings \(\varphi_1,\varphi_2,\varphi_3,\varphi_4\) defined as follows:
 \begin{itemize}
    \item Let \(\varphi_1:\sg(\Sigma_1)\times\sg(\Sigma_1)\to \ut(2,\N)\) be defined by
    \begin{align*}
        a&\mapsto\begin{pmatrix}2&0\\0&1\end{pmatrix}, & c\mapsto\begin{pmatrix}1&0\\0&2\end{pmatrix}
    \end{align*}
     \item Let \(\varphi_2:\fg(\Sigma_2)\times\fg(\Sigma_1) \to \Q^{2\times2}\) be defined by 
 \begin{align*}
     a&\mapsto\begin{pmatrix}1&2\\0&1\end{pmatrix}, & b&\mapsto\begin{pmatrix}1&0\\2&1\end{pmatrix}, & c\mapsto\begin{pmatrix}2&0\\0&2\end{pmatrix}.
 \end{align*}
 \item Let \(\varphi_3:\sg(\Sigma_2)\times\fg(\Sigma_1) \to \Z^{3\times3}\) be defined by 
 \begin{align*}
     a&\mapsto\begin{pmatrix}2&0&0\\0&1&0\\0&0&1\end{pmatrix}, & b&\mapsto\begin{pmatrix}2&1&0\\0&1&0\\0&0&1\end{pmatrix}, & c\mapsto\begin{pmatrix}1&0&0\\0&1&2\\0&0&1\end{pmatrix}.
 \end{align*}
 \item Let \(\varphi_4:\fg(\Sigma_2)\times\sg(\Sigma_1) \to \Z^{2\times2}\) be defined by 
 \begin{align*}
     a&\mapsto\begin{pmatrix}1&2\\0&1\end{pmatrix}, & b&\mapsto\begin{pmatrix}1&0\\2&1\end{pmatrix}, & c\mapsto\begin{pmatrix}2&0\\0&2\end{pmatrix}.
 \end{align*}
 \end{itemize}
Consider first the mapping \(\varphi_1\).
The images are diagonal matrices that have two 1-by-1 blocks for counting the number of letters seen.

It is straightforward to see that mappings \(\varphi_2,\varphi_3\) and \(\varphi_4\) are injective morphisms.
Indeed, they all use well-known embeddings from the literature for the first component.
In the first and third mappings, the determinant is additionally used to count the number of unary letters.
Because of this, in the first mapping, the domain is \(\Q\) as the image of \(c^{-1}\) does not exist over integers.
While in the second mapping, the block \(\begin{psmallmatrix}
    1&2\\0&1
\end{psmallmatrix}\) is used to encode words over the unary alphabet allowing the domain to be \(\Z\). \qed
\end{proof}

\embeddingnonexistence*
\begin{proof}
We prove the final three cases.
Let us first consider the fourth and fifth cases.
The proofs follow the proof of non-existence of an embedding from \(\sg(\Sigma_2)\times\sg(\Sigma_2)\) into \(\SL(3,\Z)\) found in \cite{KNP18}.
Intuitively, the only difference in the fourth case and the result of \cite{KNP18} is that we extend the codomain with elements in \(\Z\) to the codomain with elements in \(\Oc_d\).
This also applies to the fifth case.
Additionally, we move the explicit condition on determinants of matrices to implicit condition on the domain.
Indeed, while we do not restrict ourselves to images with determinant one, the same constraint applies to the first component since letters have inverses and for the image matrix to have inverse, the determinant has to be a unit.
After this shift of a constraint, the proof of \cite{KNP18} is easily adapted to prove our result.

Let us consider the third case, that is, the non-existence of an embedding from \(\fg(\Sigma_2)\times \sg(\Sigma_2)\) into \(\SL(2,\Oc_d)\).
Assume to the contrary that such an embedding exists.
Let \(\Sigma_2=\{a,b\}\) and \(\Sigma_1=\{c\}\).
Then \(\varphi:\fg(\Sigma_2)\times\sg(\Sigma_1) \to \SL(2,\Oc_d)\) maps the generators as follows:
\begin{equation*}
    (a,\varepsilon)\mapsto \begin{pmatrix}
        a_1&a_2\\a_3&a_4
    \end{pmatrix}, \hspace{0.2em} 
    (b,\varepsilon)\mapsto \begin{pmatrix}
        b_1&b_2\\b_3&b_4
    \end{pmatrix},  \hspace{0.2em}
    (\varepsilon,c)\mapsto \begin{pmatrix}
        c_1&c_2\\c_3&c_4
    \end{pmatrix}.
\end{equation*}
Without loss of generality we can assume that \(\begin{psmallmatrix}
        a_1&a_2\\a_3&a_4
    \end{psmallmatrix}\) is in the Jordan normal form as conjugating by invertible matrices does not affect the injectivity of a mapping.
Note that \(\Oc_d\) is not algebraically closed, thus the resulting matrices have elements over \(\C\) rather than \(\Oc_d\).
There are three potential Jordan normal forms for 2-by-2 matrices.
Namely,
\begin{equation*}
    (a,\varepsilon)\mapsto \begin{pmatrix}
        \lambda&0\\0&\lambda
    \end{pmatrix}, \begin{pmatrix}
        \lambda&1\\0&\lambda
    \end{pmatrix} \text{ or } \begin{pmatrix}
        \lambda_1&0\\0&\lambda_2
    \end{pmatrix},
\end{equation*}
where the first two forms are possible if the matrix has repeated eigenvalues and the final form is when there are two distinct eigenvalues.
We will rule out each form, thus proving that the embedding does not exist.
Recall that we have the following relations:
\begin{equation}
\begin{aligned}\label{eq:relsappendix}
    \varphi((a,\varepsilon)(b,\varepsilon))&\neq\varphi((b,\varepsilon)(a,\varepsilon)) \\
    \varphi((a,\varepsilon)(\varepsilon,c)) &= \varphi((\varepsilon,c)(a,\varepsilon)) \\
    \varphi((b,\varepsilon)(\varepsilon,c)) &= \varphi((\varepsilon,c)(b,\varepsilon)).
\end{aligned}
\end{equation}

Consider the first form, \((a,\varepsilon)\mapsto \begin{psmallmatrix}
        \lambda&0\\0&\lambda
    \end{psmallmatrix}\).
As this matrix is diagonal, it commutes with every matrix and in particular \(\varphi((a,\varepsilon)(b,\varepsilon))=\varphi((b,\varepsilon)(a,\varepsilon))\), which contradicts the first relation in \eqref{eq:relsappendix}.
The second form, \((a,\varepsilon)\mapsto \begin{psmallmatrix}
        \lambda&1\\0&\lambda
    \end{psmallmatrix}\), is also straightforward.
It is easy to see that matrices of this form commute only with matrices of the form \(\begin{psmallmatrix}
    x&y\\0&x
\end{psmallmatrix}\).
Hence, \(\varphi(\varepsilon,c)\) has to be of this form.
But now, for \(\varphi((b,\varepsilon)(\varepsilon,c))= \varphi((\varepsilon,c)(b,\varepsilon))\) to hold, we observe that \(\varphi(b,\varepsilon)=\begin{psmallmatrix}
    u& v \\ w &u
\end{psmallmatrix}\) with either \(w=0\) or \(y=0\).
In the first subcase, \((b,\varepsilon)\mapsto \begin{psmallmatrix}
        u&v\\0&u
    \end{psmallmatrix}\), for some \(u,v\) which is exactly the type of matrices that commute with \(\begin{psmallmatrix}
        \lambda&1\\0&\lambda
    \end{psmallmatrix}\).
    This contradicts the first relation of \eqref{eq:relsappendix}.
In the second subcase, \((\varepsilon,c)\mapsto \begin{psmallmatrix}
        x&0\\0&x
    \end{psmallmatrix}\), and since the conjugation by an invertible matrix does not affect the determinant, \(x=\pm1\) holds.
This implies that \(\varphi(\varepsilon,c)^3=\varphi(\varepsilon,c)\) contradicting our assumption that \(\varphi\) was an embedding.

Let us consider the last Jordan normal form.
Assume that \((a,\varepsilon)\mapsto \begin{psmallmatrix}
        \lambda_1&0\\0&\lambda_2
    \end{psmallmatrix}\).
Consider the products \(\varphi((a,\varepsilon)(\varepsilon,c))\) and \(\varphi((\varepsilon,c)(a,\varepsilon))\):
\begin{align*}
    \varphi((a,\varepsilon)(\varepsilon,c))&= \begin{pmatrix}
        \lambda_1c_1&\lambda_1c_2\\ \lambda_2c_3&\lambda_2c_4
    \end{pmatrix}\\
    \varphi((\varepsilon,c)(a,\varepsilon))&= \begin{pmatrix}
        \lambda_1c_1&\lambda_2c_2\\ \lambda_1c_3&\lambda_2c_4
    \end{pmatrix}.
\end{align*}
For these two matrices to be equal, we observe that \(c_2=c_3=0\).
Consider then the products \(\varphi((b,\varepsilon)(\varepsilon,c))\) and \(\varphi((\varepsilon,c)(b,\varepsilon))\):
\begin{align*}
    \varphi((b,\varepsilon)(\varepsilon,c))&= \begin{pmatrix}
        b_1c_1&b_2c_4\\ b_3c_1&b_4c_4
    \end{pmatrix}\\
    \varphi((\varepsilon,c)(b,\varepsilon))&= \begin{pmatrix}
        b_1c_1&b_2c_1\\ b_3c_4&b_4c_4
    \end{pmatrix}.
\end{align*}

Assume first that \(c_1\neq c_4\).
We require that the equality \(b_2=b_3=0\) holds.
However, we now violate the first relation of \eqref{eq:relsappendix} as \(\varphi((a,\varepsilon)(b,\varepsilon))=\begin{psmallmatrix}
    \lambda_1b_1&0\\0&\lambda_2b_4
\end{psmallmatrix}=\varphi((b,\varepsilon)(a,\varepsilon))\).
We conclude that \(c_1=c_4\) and \((\varepsilon,c)\mapsto \begin{psmallmatrix}
        c_1&0\\0&c_1
    \end{psmallmatrix}\).
Recall that this matrix is not necessarily in \(\SL(2,\Oc_d)\) as it was conjugated by some invertible matrix in order to turn \(\varphi(a,\varepsilon)\) into the Jordan normal form.
However, as before, the determinant remains 1 and thus \(c_1=\pm1\).
In either case, \(\varphi(\varepsilon,c)^3=\varphi(\varepsilon,c)\) which contradicts the assumption that \(\varphi\) was an embedding.
\qed
\end{proof}

\end{document}